\documentclass[11pt]{article}
\usepackage{fullpage}
\usepackage{anyfontsize}
\usepackage{amsmath}

\newcommand{\titlestyle}[1]{\rmfamily\mdseries\upshape \color{MediumVioletRed}{{#1}}}
\usepackage{titlesec}
\titleformat{\section}{\raggedleft\rmfamily\Large\color{MediumVioletRed}}{\thesection}{1em}{}
\titleformat{\subsection}{\raggedleft\rmfamily\large\color{MediumVioletRed}}{\thesubsection}{1em}{}

\usepackage{amsthm}
\let\oldproofname=\proofname
\renewcommand{\proofname}{\mbox{\sffamily\upshape\bfseries \color{Purple}{\oldproofname}}}

\usepackage{mathtools}
\usepackage{thmtools}
\declaretheoremstyle[
    headfont=\sffamily\scshape\bfseries \color{Purple},
    bodyfont=\rmfamily\itshape, 
    postheadspace=1em,
]{mytheoremstyle}
\declaretheorem{theorem}[
    style=mytheoremstyle,
    numberwithin=,
]
\declaretheorem{definition, lemma, 
    proposition, corollary, 
    example, observation}[
        style=mytheoremstyle,
        sibling=theorem, 
]

\usepackage{amssymb} 
\usepackage{cite}
\usepackage[colorlinks,citecolor=Fuchsia]{hyperref}
\usepackage[dvipsnames,svgnames]{xcolor}
\usepackage{algorithm}
\usepackage{algpseudocodex}
\tikzset{algpxIndentLine/.style=thick}
\usepackage{enumitem}
\usepackage{multirow}
\usepackage{multicol}
\usepackage{tabu}
\usepackage{caption}
\usepackage{longtable}
\newcolumntype{L}{>{$}l<{$}} 
\newcolumntype{C}{>{$}c<{$}} 
\newcolumntype{R}{>{$}r<{$}} 
\usepackage{tikz}
\usetikzlibrary{calc}
\usetikzlibrary{arrows}
\usetikzlibrary{positioning}
\usetikzlibrary{fit}
\usetikzlibrary{shapes}

\usepackage{subcaption}
\DeclareCaptionFormat{custom}{
    \mbox{\sffamily\upshape\bfseries \color{Crimson}{#1 #2}} {#3}
}
\usepackage{ebproof}
\ebproofset{rule thickness=0.7pt}
\newtheorem{obs}[theorem]{Observation}

\newcommand{\sfemph}[1]{\mbox{\sffamily\upshape\mdseries{#1}}}
\newcommand{\defemph}[1]{\mbox{\rmfamily\itshape \color{Blue} {#1}}}

\newcommand{\dom}{{\sf dom}}

\newcommand{\embrace}[1]{\lbrace{#1}\rbrace}

\newcommand{\names}{\sfemph{N}}
\newcommand{\Terms}{\sfemph{T}}
\newcommand{\groundterms}{\sfemph{G}}
\newcommand{\vars}{\sfemph{V}}
\newcommand{\ag}{\sfemph{A}}
\newcommand{\keys}{\sfemph{K}}

\newcommand{\varsof}{{\sf vars}}

\newcommand{\subterms}[1]{{\sf st}({#1})}

\newcommand{\size}[1]{\Vert{#1}\Vert}
\newcommand{\func}{{\sf f}}

\newcommand{\pairfn}[2]{({#1},{#2})}
\newcommand{\sencfn}[2]{\left\lbrace {#1}\right\rbrace_{{#2}}}
\newcommand{\pkfn}[1]{\textsf{pk}({#1})}
\newcommand{\aencfn}[2]{\left\lbrace {#1}\right\rbrace_{{#2}}}

\newcommand{\nf}[1]{{\sf nf}\left({#1}\right)}
\newcommand{\xor}{\oplus}
\newcommand{\Xor}[1]{\bigoplus{#1}}
\newcommand{\factors}{{\sf fac}}
\newcommand{\scxor}{\textsc{xor}}
\newcommand{\occ}[2]{\textsf{occ}_{{#1}}({#2})}

\newcommand{\rnrule}{{\sf r}}
\newcommand{\rnax}{{\sf ax}}
\newcommand{\rnpk}{{\sf pk}}
\newcommand{\rnpair}{{\sf pair}}
\newcommand{\rnsplit}{{\sf split}}

\newcommand{\rnsenc}{{\sf senc}}
\newcommand{\rnsdec}{{\sf sdec}}
\newcommand{\rnaenc}{{\sf aenc}}
\newcommand{\rnadec}{{\sf adec}}
\newcommand{\rnxor}{{\sf xor}}
\newcommand{\rnxorc}{{\sf xor}_{c}}
\newcommand{\rnxord}{{\sf xor}_{d}}

\newcommand{\DYderives}{\vdash_{\textit{dy}}}
\newcommand{\DYnderives}{\nvdash_{\textit{dy}}}
\newcommand{\axiomsof}{{\sf axioms}}
\newcommand{\concof}{{\sf conc}}
\newcommand{\lrof}{{\sf lr}}
\newcommand{\termsof}{{\sf terms}}
\newcommand{\secret}{{\sf secret}}

\newcommand{\onestep}{\textsf{one-step}}
\newcommand{\stname}{{\sf st}}

\newcommand{\prot}{\sfemph{Prot}}
\newcommand{\role}{\sfemph{R}}
\newcommand{\agvarsof}{{\sf vars}_{a}}
\newcommand{\intvarsof}{{\sf vars}_{i}}

\newcommand{\recsend}[2]{{#1}\!\Rightarrow\!{#2}}

\newcommand{\ses}{\mbox{\sffamily\upshape\mdseries ses}}
\newcommand{\sessions}{\mathcal{S}}
\newcommand{\runs}{\mathcal{R}}

\newcommand{\constst}{\sfemph{C}}
\newcommand{\stnonvars}{\sfemph{D}}
\newcommand{\intsub}{\sigma}
\newcommand{\vintsub}{\sigma^{*}}
\newcommand{\honsub}{\tau}
\newcommand{\intconst}{\constst\intsub}

\newcommand{\nfintconst}{\nf{\constst\intsub}}
\newcommand{\nfintnonvars}{\nf{\stnonvars\intsub}}

\newcommand{\zap}[1]{{\sf zap}\left({#1}\right)}

\begin{document}
\title{\titlestyle{Protocol insecurity with finitely many sessions and \scxor{}}}
\author{
    R Ramanujam\\ 
    Azim Premji University, Bengaluru (Visiting)\\
    \texttt{ramanujam.r@apu.edu.in}
    \and
    Vaishnavi Sundararajan\\
    Indian Institute of Technology, Delhi\\
    \texttt{vaishnavi@cse.iitd.ac.in}
    \and
    S P Suresh\\
    Chennai Mathematical Institute\\
    \texttt{spsuresh@cmi.ac.in}
}
\date{}
\maketitle
\raggedbottom
\begin{abstract}
    We present a different proof of the insecurity problem for \scxor{}, solved in~\cite{CKRT05}. Our proof uses the notion of typed terms and well-typed proofs, and removes a restriction on the class of protocols to which the~\cite{CKRT05} proof applies, by introducing a slightly different (but very natural) notion of protocols, where honest agent sends are derivable from previous receives in the same session. 
\end{abstract}

\section{Dolev-Yao with \scxor}\label{sec:xor}

We begin with a set $\names$ of names (atomic terms, with no further structure), and a set of variables $\vars$. We denote by $\ag \subseteq \names$ the set of agents, with $I \in \ag$ being the malicious intruder. The set of {keys} is denoted by $\keys \subseteq \names$. The set of terms, denoted by $\Terms$, is given by
\[
    t, t' \in \Terms\quad \Coloneq \quad x\ \mid\ m\ \mid\ \pkfn{k}\ \mid\ \pairfn{t}{t'}\ \mid\ \sencfn{t}{t'}\ \mid\ \aencfn{t}{\pkfn{k}}\ \mid\ \textstyle{\Xor{M}} 
\]
where $x \in \vars$, $m \in \names$, $k \in \keys$, $t, t' \in \Terms$, and $M$ is a multiset of terms. We write $\Xor{\embrace{t_{1}, \ldots, t_{n}}}$ in the more readable form $t_{1} \xor \cdots \xor t_{n}$. We will assume that $0$ and $\secret$ are two particular terms in $\names$. A term is \defemph{non-standard} if it is of the form $\Xor{M}$, and \defemph{standard} otherwise. We will assume that $\Xor{\{t\}} \coloneq t$ for a standard term $t$, and $\Xor{\emptyset} \coloneq 0$. 

For a multiset $M$, we define $\occ{c}{M}$ to be the number of times $c$ occurs in $M$. The \defemph{factors} of a term $t$, denoted $\factors(t)$, is a multiset defined as follows (with $\Cup$ denoting multiset union): 
\[
    \factors(t) \coloneq
    \begin{cases}
        \embrace{t} & \text{ if } t \text{ is standard} \\
        \Cup_{u \in M}\factors(u) & \text{ if } t = \Xor{M}
    \end{cases}
\]
The $\scxor$ operator is considered to be associative, commutative, and nilpotent ($x \xor x = 0$), with $0$ as the unit element. These rules can be captured in our model by defining a normal form for every term $t$, which is a canonical term from the equivalence class of $t$ in the equivalence relation defined by the $\xor$ identities. The \defemph{normal form} of $t$ -- denoted $\nf{t}$, and the \defemph{subterms} of $t$ -- denoted $\subterms{t}$, are defined as in Table~\ref{tab:nfst}. For a set of terms $X$, $\nf{X} := \embrace{\nf{t}\ \mid\ t \in X}$, and $\subterms{X} \coloneq \bigcup_{t \in X}~\subterms{t}$. A term is \defemph{normalized} if $\nf{t} = t$. The \defemph{size}\footnote{This is technically the \emph{dag-size} of $t$ (and $X$), as we count multiple occurrences of the same subterm only once.} of a term $t$ is $\size{t} \coloneq |\subterms{t}|$. For $X \subseteq \Terms$, we define $\size{X} \coloneq |\subterms{X}|$.  

\begin{table}
    \centering 
    \tabulinesep=0.8mm
    \setlength{\tabcolsep}{0.2em}
    \begin{math}
        \begin{tabu}{|c||c|c|}
            \hline
            t & \subterms{t} & \nf{t} \\ 
            \hline
            a & \{a\} & a \\
            \pkfn{k} & \{\pkfn{k}, k\} & \pkfn{k} \\
            \func(u,v) & \{\func(u,v)\} \cup \subterms{u} \cup \subterms{v} & \func(\nf{u}, \nf{v}) \\ 
            r & \{r\} \cup \textstyle{\bigcup_{u \in \factors(r)}}\subterms{u} & \textstyle{\Xor{}\left\{v\ \middle\vert\ \left({\sum_{u, \nf{u} = v}~\occ{u}{\factors(r)}}\right)\text{ is odd}\right\}} \\
            \hline
        \end{tabu}
    \end{math}
    \caption{\mbox{\sffamily\bfseries Normal forms and subterms}. $a \in \vars\cup\names$, $k \in \keys$, $\func(u,v)$ is standard, and $r$ is non-standard. In the last row, $\nf{r}$ is an $\scxor$ of a \emph{set} of terms.}
    \label{tab:nfst}
\end{table}

The set of variables appearing in $t$ is denoted by $\varsof(t)$. The set of \defemph{ground} terms, $\groundterms$, are those without variables. A substitution is a finite partial map $\sigma: \vars \to \Terms$. A substitution $\sigma$ is ground if $\sigma(x) \in \groundterms$ for all $x \in \dom(\sigma)$, and it is normalized if $\sigma(x)$ is normalized for all $x \in \dom(\sigma)$. For a substitution $\sigma$ and term $t$, we define $t\sigma$ as follows (where $x \in \dom(\sigma)$, $y \notin \dom(\sigma)$, $a \in \names$, and $\func$ is any cryptographic operator):
\[
    x\sigma \coloneq \sigma(x)  \quad\qquad y\sigma \coloneq y  \quad\qquad a\sigma \coloneq a  \quad\qquad \func(t_{1}, \ldots, t_{n})\sigma \coloneq \func(t_{1}\sigma, \ldots, t_{n}\sigma). 
\]
For $X \subseteq \Terms$, $X\sigma \coloneq \embrace{t\sigma \mid t \in X}$. The \emph{size of a substitution} $\sigma$ is $\size{\sigma} \coloneq |\subterms{\{\sigma(x) \mid x \in \dom(\sigma)\}}|$.   

\begin{table}
    \centering
    \tabulinesep=2mm
    \setlength{\tabcolsep}{0.6em}
    \begin{tabu}{|c|c|c|c|c|}
        \hline
        \begin{prooftree}
            \infer0[\rnax]{\ t\ }
        \end{prooftree}
        &
        \begin{prooftree}
            \hypo{\pairfn{t_{1}}{t_{2}}}
            \infer1[\rnsplit]{t_{i}}
        \end{prooftree}
        &
        \begin{prooftree}
            \hypo{\sencfn{u}{v}}
            \hypo{v}
            \infer2[\rnsdec]{u}
        \end{prooftree}
        & 
        \begin{prooftree}
            \hypo{\aencfn{u}{\pkfn{k}}}
            \hypo{k}
            \infer2[\rnadec]{u}
        \end{prooftree}
        &
        \multirow{2}{*}{
            \begin{prooftree}
                \hypo{t_{1}} \hypo{\ldots} \hypo{t_{n}}
                \infer3[\rnxor]{\nf{t_{1} \xor \cdots \xor t_{n}}} 
            \end{prooftree}
        }
        \\
        \cline{1-4} 
        \begin{prooftree}
            \hypo{k} 
            \infer1[\rnpk]{\pkfn{k}}
        \end{prooftree}
        &
        \begin{prooftree}
            \hypo{u}
            \hypo{v}
            \infer2[\rnpair]{\pairfn{u}{v}}
        \end{prooftree}
        &
        \begin{prooftree}
            \hypo{u}
            \hypo{v}
            \infer2[\rnsenc]{\sencfn{u}{v}}
        \end{prooftree}
        &
        \begin{prooftree}
            \hypo{u}
            \hypo{\pkfn{k}}
            \infer2[\rnaenc]{\aencfn{u}{\pkfn{k}}}
        \end{prooftree}
        &
        \\
        \hline
    \end{tabu}
    \caption{\mbox{\sffamily\bfseries Derivation rules for terms}. The terms above the horizontal line in each rule are \defemph{premises} and the one below the line is the \defemph{conclusion}. The first row has \defemph{destructor rules}, and the second row has \defemph{constructor rules}. If the conclusion in $\rnxor$ is standard, we will sometimes refer to it as $\rnxord$ and count it among the destructor rule, else we sometimes refer to it as $\rnxorc$ and count it among the constructor rules.}
    \label{tab:termalgtab}
\end{table}

\begin{definition}[Derivations]
    A \defemph{derivation} (or \defemph{proof\/}) is a tree $\pi$ with the following properties:
    \begin{itemize}
        \item Each node is labelled with a pair $(t, \rnrule)$ where $t$ is a normalized term and $\rnrule$ is one of the rules from Table~\ref{tab:termalgtab}. 
        \item Each leaf is labelled with a pair $(u, \rnax)$.
        \item If a node is labelled $(t, \rnrule)$ and its children are labelled $(t_{1}, \rnrule_{1}), \ldots, (t_{n}, \rnrule_{n})$ then there is an instance of $\rnrule$ with premises $t_{1}, \ldots, t_{n}$ and conclusion $t$. 
    \end{itemize} 
    If $(t, \rnrule)$ labels the root of a derivation $\pi$, then the \emph{last rule} of $\pi$ (denoted $\lrof(\pi)$) is $\rnrule$, and the \defemph{conclusion} of $\pi$ (denoted $\concof(\pi)$) is $t$. The \defemph{axioms} of $\pi$, denoted $\axiomsof(\pi)$ is the set $\embrace{u \mid (u,\rnax) \text{ labels some leaf of }\pi}$. An \defemph{occurrence of a rule} $\rnrule$ in $\pi$ is just a node that is labelled $(t, \rnrule)$ for some $t$. A \defemph{subproof} of $\pi$ is just a subtree.
    
    If $\axiomsof(\pi) \subseteq X$ and $\concof(\pi) = t$, then we say that $\pi$ proves $X \vdash t$. If there is a proof of $X \vdash t$, we also write $X \DYderives t$. We write $X \DYnderives t$ to mean that it is not the case that $X \DYderives t$. We always consider normalized $X \cup \embrace{t}$ when we talk about $X \vdash t$.
\end{definition}

\begin{definition}[Normal proofs]
    A proof $\pi$ is said to be \defemph{normal} if the following conditions hold: 
    \begin{enumerate}
        \item $\pi$ has no occurrence of the $\rnsplit, \rnsdec, \rnadec$ rules whose leftmost premise is the conclusion of a constructor.
        \item In all occurrences of the $\rnxor$ rule in $\pi$, 
        \begin{enumerate}
            \item no two premises are the same, and
            \item no premise is the same as the conclusion, and  
            \item no premise is the conclusion of $\rnxor$.
        \end{enumerate}
    \end{enumerate}
\end{definition}

\begin{theorem}[Normalization for $\DYderives$] 
    If $X \DYderives t$, there is a normal proof of $X \vdash t$. 
\end{theorem}
\begin{proof}
    The offending patterns can be rewritten, each time decreasing the size of the proof (number of nodes in the proof tree). As an example of ensuring the first condition, suppose there is a subproof that looks as follows: 
    \[
        \begin{prooftree}
            \hypo{}
            \ellipsis{$\pi_{1}$}{\ t\ }
            \hypo{}
            \ellipsis{$\pi_{2}$}{\ \pkfn{k}\ }
            \infer2[$\rnaenc$]{\aencfn{t}{\pkfn{k}}}
            \hypo{}
            \ellipsis{$\pi_{3}$}{\ k\ }
            \infer2[$\rnadec$]{t} 
        \end{prooftree}
    \]
    It can simply be replaced by $\pi_{1}$, yielding a smaller proof. 
    
    As to the second condition, suppose $\pi$ is a subproof with last rule $\rnxor$. If two of its immediate subproofs have the same conclusion, then they can both be removed, without affecting the conclusion, and we have a smaller proof.

    If one of the immediate subproofs $\delta$ has the same conclusion as $\pi$, then $\pi$ can be replaced with $\delta$, a smaller proof. 
    
    Suppose that $\pi$ has the following structure (with $t_{i} = \nf{u_{1} \xor \cdots \xor u_{k}}$ and $t = \nf{t_{1} \xor \cdots \xor t_{n}}$): 
    \[
        \begin{prooftree}[separation=2em]
            \hypo{}\ellipsis{$\pi_{1}$}{t_{1}}
            \hypo{\cdots} 
            \hypo{}\ellipsis{$\pi_{i-1}$}{t_{i-1}}
            \hypo{}\ellipsis{$\delta_{1}$}{u_{1}}
            \hypo{\cdots} 
            \hypo{}\ellipsis{$\delta_{k}$}{u_{k}}
            \infer3[\rnxor]{t_{i}}
            \hypo{}\ellipsis{$\pi_{i+1}$}{t_{i+1}}
            \hypo{\cdots} 
            \hypo{}\ellipsis{$\pi_{n}$}{t_{n}}
            \infer7[\rnxor]{t}
        \end{prooftree}
    \]
    It is easy to see that $t = \nf{t_{1} \xor \cdots t_{i-1} \xor u_{1} \xor \cdots \xor u_{k} \xor t_{i+1} \xor \cdots \xor t_{n}}$, and hence $\pi$ can be replaced by the following smaller proof. 
    \[
        \begin{prooftree}[separation=2em]
            \hypo{}\ellipsis{$\pi_{1}$}{t_{1}}
            \hypo{\cdots} 
            \hypo{}\ellipsis{$\pi_{i-1}$}{t_{i-1}}
            \hypo{}\ellipsis{$\delta_{1}$}{u_{1}}
            \hypo{\cdots} 
            \hypo{}\ellipsis{$\delta_{k}$}{u_{k}}
            \hypo{}\ellipsis{$\pi_{i+1}$}{t_{i+1}}
            \hypo{\cdots} 
            \hypo{}\ellipsis{$\pi_{n}$}{t_{n}}
            \infer9[\rnxor]{t}
        \end{prooftree}
    \]
    Thus any proof can be converted to a normal proof by applying these proof transformations repeatedly till they can no longer be applied.
\end{proof}

\begin{obs}\label{obs:lrxor}
    Suppose $\pi$ is a normal proof with last rule $\rnxor$ and conclusion $t$. Let the immediate subproofs be $\pi_{1}, \ldots, \pi_{n}$, with conclusions $t_{1}, \ldots, t_{n}$ respectively. Then the following hold: 
    \begin{enumerate}
        \item If $t_{i}$ is non-standard, then $\lrof(\pi_{i})$ is a destructor rule other than $\rnxord$. 
        \item If $t_{i}$ is standard, then 
        \begin{enumerate}
            \item If $\lrof(\pi)$ is $\rnxord$, then $t_{i}$ is a factor of some non-standard $t_{j}$. Furthermore, $t$ is a factor of some non-standard $t_{j}$.  
            \item If $\lrof(\pi)$ is $\rnxorc$, then $t_{i}$ is either a factor of $t$ or a factor of some non-standard $t_{j}$.  
        \end{enumerate}
    \end{enumerate}
\end{obs}
\begin{proof}
    \phantom{a}
    \begin{enumerate}
        \item Suppose $t_{i}$ is non-standard. By normality of $\pi$, $\lrof(\pi_{i}) \neq \rnxor$. Since $t_{i}$ is non-standard, $\lrof(\pi)$ is not in $\{\rnpair, \rnsenc, \rnaenc\}$. Thus $\lrof(\pi_{i})$ has to be a destructor other than $\rnxord$. 
        \item Suppose $t_{i}$ is standard. 
        \begin{enumerate}
            \item By normality, $t_{i} \neq t$ and $t_{i} \neq t_{j}$ for any $j \neq i$. Since $t_{i}$ disappears as a result of the $\xor$ operator, this means that $t_{i} \in \factors(t_{j})$ for some non-standard $t_{j}$. 
            
            Since $\lrof(\pi)$ is $\rnxord$, $t$ is a standard term equal to $\nf{t_{1} \xor \cdots \xor t_{n}}$, and this can happen only if $t \in \factors(t_{j})$ for some non-standard $t_{j}$ ($t_{j}$ cannot be standard, for then it would be the same as $t$, violating the normality of $\pi$).

            \item A standard $t_{i}$, if it is not a factor of $t$, must have been cancelled by a factor of a non-standard $t_{j}$, and thus we have the desired conclusion. 
            \qedhere
        \end{enumerate}
    \end{enumerate}
\end{proof}

\begin{theorem}[Subterm property]
    Let $\pi$ be a normal proof with $\axiomsof(\pi) \subseteq X$ and $\concof(\pi) = t$. Then $\termsof(\pi) \subseteq \subterms{X \cup \embrace{t}}$. Further, if $\lrof(\pi)$ is a destructor, then $\termsof(\pi) \subseteq \subterms{X}$. 
\end{theorem}
\begin{proof}
    Assume that the statement of the theorem holds for all subproofs of $\pi$.\footnote{Note that $\axiomsof(\delta) \subseteq X$ for all subproofs $\delta$ of $\pi$.} Letting $\lrof(\pi) = \rnrule$, the following cases arise:
    \begin{description}
        \item[$\rnrule = \rnax$:] Then $\termsof(\pi) \subseteq \axiomsof(\pi) \subseteq X \subseteq \subterms{X}$.
        \item[$\rnrule = \rnsplit$ or $\rnadec$:] Similar to the $\rnsdec$ case dealt with later.  
        \item[$\rnrule = \rnpk, \rnpair$ or $\rnaenc$:] Similar to the $\rnsenc$ case dealt with later. 
        
        \item[$\rnrule = \rnsdec$:] Let $\delta$ and $\delta'$ be the immediate subproofs, with conclusions $\sencfn{t}{u}$ and $u$, respectively. Since $\pi$ is normal, $\delta$ does not end in a constructor rule, so the stronger induction statement holds -- $\termsof(\delta) \subseteq \subterms{X}$. In particular, $\sencfn{t}{u} \in \subterms{X}$, and hence $t, u \in \subterms{X}$. By IH for $\delta'$, $\termsof(\delta') \subseteq \subterms{X \cup \embrace{u}} \subseteq \subterms{X}$. It follows that $\termsof(\pi) = \termsof(\delta) \cup \termsof(\delta') \cup \embrace{t} \subseteq \subterms{X}$. 
        
        \item[$\rnrule = \rnsenc$:] Let $t = \sencfn{u}{v}$, and let $\delta$ and $\delta'$ be the immediate subproofs, with conclusions $u$ and $v$. By IH, $\termsof(\delta) \subseteq \subterms{X \cup \embrace{u}} \subseteq \subterms{X \cup \embrace{t}}$ and $\termsof(\delta') \subseteq \subterms{X \cup \embrace{v}} \subseteq \subterms{X \cup \embrace{t}}$, and it follows that $\termsof(\pi) = \termsof(\delta) \cup \termsof(\delta') \cup \embrace{t} \subseteq \subterms{X \cup \embrace{t}}$. 
        
        \item[$\rnrule = \rnxord$:] Let $\delta$ be any immediate subproof of $\pi$. If its conclusion is non-standard, then it ends in a destructor rule (by Observation~\ref{obs:lrxor}). By IH, $\termsof(\delta) \subseteq \subterms{X}$ for such a $\delta$. In particular, all non-standard premises of $\rnrule$ are in $\subterms{X}$. 
        
        If the conclusion $u$ of $\delta$ is standard, then by Observation~\ref{obs:lrxor}, $u$ is a subterm of some non-standard premise of $\rnrule$, which itself is a subterm of $X$, and by IH, $\termsof(\delta) \subseteq \subterms{X \cup \embrace{u}} \subseteq \subterms{X}$.

        Finally, $t$ itself is a factor of some non-standard premise (again by Observation~\ref{obs:lrxor}), and hence $t \in \subterms{X}$. Thus we see that $\termsof(\pi) \subseteq \subterms{X}$. 
        
        \item[$\rnrule = \rnxorc$:] By arguing as in the above case, we see that for all subproofs $\delta$ with a non-standard term as conclusion, $\termsof(\delta) \subseteq \subterms{X}$. For subproofs $\delta$ with a standard term $u$ as conclusion, $u$ is either a subterm of one of the non-standard premises, or a subterm of $t$ (by Observation~\ref{obs:lrxor}). Thus we see that $\termsof(\pi) \subseteq \subterms{X \cup \embrace{t}}$. 
        \qedhere 
    \end{description}
\end{proof}

\begin{theorem} 
    There is a \textsc{ptime} algorithm to check whether $X \vdash t$, for any given normalized set $X \cup \embrace{t}$.
\end{theorem}
\begin{proof}
    Let $\stname$ be shorthand for $\subterms{X \cup \embrace{t}}$, with $|\stname| = N$. It can be easily seen that $N \leq \size{X \cup \embrace{t}}$. For any $Y \subseteq \stname$, we define 
    \[
        \onestep(Y) \coloneq Y \cup \{u \in \stname \mid u \text{ is the conclusion of a rule all of whose premises are in } Y\}.
    \] 
    Note that $Y \subseteq \onestep(Y)$. We define an infinite sequence of sets $Y_{0} \subseteq Y_{1} \subseteq \ldots$ with $Y_{0} \coloneq X$ and $Y_{i+1} \coloneq \onestep(Y_{i})$ for each $i \geq 0$. Since there are at most $N$ elements in $\stname$, we have that $Y_{i} = Y_{N}$ for all $i \geq N$. It is clear that $X \vdash t$ iff $t \in Y_{N}$. 

    To compute $\onestep(Y)$, start with $Z \coloneq Y$ and add elements to $Z$ in two stages. In the first stage, for every $M \subseteq Y$ with $|M| \leq 2$, and every $\rnrule \neq \rnxor$, if we can apply $\rnrule$ to the premises $M$ and get some $u \in \stname$ as conclusion, we add $u$ to $Z$. This stage runs in $O(N^{2})$ time. In the next stage, for each $u \in \stname$, we add $u$ to $Z$ if there is $M \subseteq Y$ such that $u = \nf{\Xor{M}}$. This check can be done as follows: we treat each term in $\stname$ as a vector over the standard terms in $\stname$ with coefficients from $\embrace{0,1}$, and use Gaussian elimination or a similar procedure to check if the vector representing $u$ is obtained as a linear combination (over $\mathbb{Z}/2$) of the vectors representing each element in $Y$. This test can be performed in time $O(N^{3})$, at worst. The $Z$ we get after the two stages -- running in time $O(N^{4})$ -- is $\onestep(Y)$. Since we compute $\onestep$ $N$ times, the overall running time is $O(N^{5})$.
\end{proof}

\section{Protocols and runs}\label{sec:prot}
\begin{definition}[Protocols]\label{def:prot}
    A \defemph{role} $\role$ is a pair $(X, \rho)$ where $X$ is a set of \emph{standard normalized terms}, and $\rho$ is a finite sequence $(r_{1}, s_{1})\ldots(r_{n}, s_{n})$ with each $r_{i}$ and $s_{i}$ a normalized term (denoting a \emph{received term} and the term \emph{sent} in response). For any role $\role = (X,\rho)$ as above, $\agvarsof(\rho) \coloneq \bigcup_{i \leq n} \left(\varsof(s_{i})\setminus\varsof(\{r_{1}, \ldots, r_{i}\})\right)$ is the set of \defemph{agent variables} of $\role$ and $\intvarsof(\rho) \coloneq \varsof(\rho) \setminus \agvarsof(\rho)$ is the set of \defemph{intruder variables} of $\role$.\footnote{Agent variables are those that first occur in a send, and are to be instantiated by an agent playing the role as part of a protocol execution. The intruder variables are the ones that first occur in a receive, and whose meaning in an execution will be fixed by the intruder.} A role $(X, \rho)$ as above is said to be \defemph{well-formed} if $\varsof(X) \subseteq \agvarsof(\rho)$ and for all $i \leq n$, $X \cup \{r_{1}, \ldots, r_{i}\} \DYderives s_{i}$.\footnote{Protocol models usually do not make explicit these ``honest agent derivability checks'', since it is assumed that all realistic protocols satisfy this requirement -- every send in a session is derivable from the previous receives in that session. The proof of Lemma~\ref{lem:sigmax-in-sigmaTp} benefits from this assumption having been made explicit.}

    A \defemph{protocol} $\prot$ is a pair $(X_{0}, \{\role_{1}, \ldots, \role_{k}\})$, where $X_{0}$ is a finite set of \emph{normalized ground terms} (representing the \defemph{intruder's initial knowledge}) and each $\role_{i}$ is a well-formed role. 
\end{definition}

\begin{definition}[Sessions of a protocol]\label{def:sessions}
    A \defemph{session} of a protocol $\prot$ is a pair $\ses = (\role, \honsub)$ where $\role = (X, \rho)$ is a role of $\prot$; $\honsub$ is a substitution with $\honsub(x) \in \names$ for $x \in \agvarsof(\role)$ and $\honsub(x) \in \vars$ for $x \in \intvarsof(\role)$; and $\rho\honsub$ consists of normalized terms only.

    Two sessions $((X_{1}, \rho_{1}), \honsub_{1})$ and $((X_{2}, \rho_{2}), \honsub_{2})$ are said to be \defemph{coherent} if $\varsof(\rho_{1}\honsub_{1}) \cap \varsof(\rho_{2}\honsub_{2}) = \emptyset$.
\end{definition}

\begin{definition}[Runs of a protocol]\label{def:protruns}
    Suppose $\prot = (X_{0}, \{\role_{1}, \ldots, \role_{k}\})$ is a protocol, and suppose that 
    \[
        \sessions = \{((X_{1}, \rho_{1}), \honsub_{1}), \ldots, ((X_{\ell}, \rho_{\ell}), \honsub_{\ell})\}
    \] 
    is a set of pairwise coherent sessions of $\prot$.
    
    The set of \defemph{runs generated by $\sessions$}, denoted $\runs(\sessions)$, consists of all pairs $(\xi, \intsub)$ s.t.:
    \begin{itemize}
		\item $\xi = (r_{1},s_{1}) \cdots (r_{n},s_{n})$ is a prefix of an interleaving of $\{\rho_{1}\honsub_{1}, \ldots, \rho_{\ell}\honsub_{\ell}\}$, and written 
        \[\recsend{r_{1}}{s_{1}}\ \cdots\ \recsend{r_{n}}{s_{n}}.\]
        \item $\intsub$ is a normalized ground substitution with $\dom(\intsub) = \varsof(\xi)$.
        \item For $0 \leq i \leq n$, letting $X_{i} = X_{0} \cup \{s_{1}, \ldots,s_{i}\}$, we have:
        \[
            \nf{X_{i-1}\intsub} \DYderives \nf{r_{i}\intsub}. 
        \]
    \end{itemize}
    We denote by  $\runs(\prot)$ the set of all $(\xi, \intsub) \in \runs(\sessions)$ for all finite sets $\sessions$ of pairwise coherent sessions of $\prot$. If $|\sessions| = K$ and $(\xi, \intsub) \in \runs(\sessions)$, we say that it is a \defemph{$K$-bounded run} of $\prot$.  
    
    If $\xi$ is a run ($K$-bounded run) as above, and $\nf{X_{n}\intsub} \DYderives \secret$, then we call $(\xi, \intsub)$ an \defemph{attack} (a $K$-bounded attack). 
\end{definition}

\begin{definition}[$K$-bounded insecurity problem (for any fixed $K \geq 0$)]\label{def:insecurity}
    Given a protocol $\prot$, check whether there exists a $K$-bounded attack $(\xi, \intsub)$ of $\prot$. 
\end{definition}
We will use ``insecurity problem'' to mean the $K$-bounded insecurity problem for some $K$.

\section{Insecurity problem for \textsc{xor}}\label{sec:insecurity}
Fix a bound $K$, a protocol $\prot$, a set of pairwise coherent sessions $\sessions = \{((X_{1},\rho_{1}),\honsub_{1}), \ldots, ((X_{K},\rho_{K}),\honsub_{K})\}$, and a run $(\xi, \intsub) \in \runs(\sessions)$ (with $\xi = \recsend{r_{1}}{s_{1}}\ \cdots\ \recsend{r_{n}}{s_{n}}$), and sets of terms $X_{0}, \ldots, X_{n}$ satisfying the conditions of Definition~\ref{def:protruns}, and s.t.\ $\nf{X_{n}\intsub} \DYderives \secret$. Note that $X_{i} \subseteq X_{i+1}$ for $i < n$. Define the following sets:
\[
    \constst \coloneq \subterms{X_{0} \cup \left(\textstyle{\bigcup_{i\leq{}K}X_{i}\honsub_{i}}\right) \cup \{r_{1}, s_{1}, \ldots, r_{n}, s_{n}, \secret\}} \qquad \stnonvars \coloneq \{t \in \constst \mid t\text{ is standard and }t \notin \vars\}
\]

\begin{obs}\label{obs:honestsends}
    For every $m \leq n$, there is a set of ground terms $Y \subseteq \stnonvars$ such that $Y \cup \{r_{1}, \ldots, r_{m}\} \DYderives s_{m}$.
\end{obs}
\begin{proof}
    Suppose $m \leq n$, and $s_{m}$ is the $m$-th send in the run $\xi$. There is an $i \leq K$ and a prefix $(r^{1}, s^{1})\ \cdots\ (r^{k}, s^{k})$ of $\rho_{i}$ such that $s_{m} = s^{k}\honsub_{i}$. Since $\rho_{i}$ is well-formed, $X_{i} \cup \{r^{1}, \ldots, r^{k}\} \vdash s^{k}$ via some proof (say $\pi$). Since $\rho_{i}\honsub_{i}$ consists only of normalized terms, $\pi$ can be transformed easily (by replacing each label $(u,\rnrule)$ with $(u\honsub_{i}, \rnrule)$) to a proof of $X_{i}\honsub_{i} \cup \{r^{1}\honsub_{i}, \ldots, r^{k}\honsub_{i}\} \vdash s^{k}\honsub_{i}$. But note that $X_{i}\honsub_{i} \subseteq \constst$ is a set of standard normalized ground terms, and $\{r^{1}\honsub_{i}, \ldots, r^{j}\honsub_{i}\} \subseteq \{r_{1}, \ldots, r_{m}\}$. Thus we see that for every sent message $s_{m}$ occurring in $\xi$, there is some set of ground terms $Y \subseteq \stnonvars$ such that $Y \cup \{r_{1}, \ldots, r_{m}\} \DYderives s_{m}$.
\end{proof}

\subsection{Types and zaps}\label{subsec:type-zap}
\begin{definition}[Typed and zappable terms]\label{def:zappable} 
    A standard normalized term $t$ is \defemph{typed} if $t \in \nfintnonvars$. A non-atomic standard term $t$ is \defemph{zappable} if $\nf{t} \notin \nfintnonvars$.  
\end{definition}

\begin{definition}[Zap]\label{def:zap}
    For a ground term $t$, we define the \defemph{zap} of $t$, denoted $\zap{t}$, as follows:
    \begin{align*}
        \zap{a} \quad &\coloneq \quad{}a &\text{if }a \in \names \\
        \zap{\func(t_{1}, \ldots, t_{n})} &\coloneq \quad 0 &\text{if }\func(t_{1}, \ldots, t_{n})\text{ is zappable} \\
        \zap{\func(t_{1}, \ldots, t_{n})} &\coloneq \quad \nf{\func(\zap{t_{1}}, \ldots, \zap{t_{n}})} & \text{otherwise}
    \end{align*}
    Note that $\zap{t}$ is normalized, by definition. For any multiset of terms $X$, let $\zap{X}$ be the \emph{multiset} $\{\,\zap{t} \mid t \in X\}$. 
\end{definition}

\begin{lemma}\label{lem:funcnormA}
    Suppose $f$ is a function on terms that satisfies the equation $f\left(\Xor{M}\right) = \nf{\Xor{f(M)}}$, for any multiset of standard normalized terms $M$. Then for any multiset $A$ of standard normalized terms, $f\left(\nf{\Xor{A}}\right) = \nf{\Xor{f(A)}}$.\footnote{Note that if $S$ is the multiset $\{s_{1}, \ldots, s_{n}\}$, then $\Xor{f(S)} = f(s_{1}) \xor \cdots \xor f(s_{n})$.} 
\end{lemma}
\begin{proof}
    Since each element of $A$ is normalized, we have $\nf{\Xor{A}} = \Xor{B}$ for a \emph{set} $B \subseteq A$ such that $\occ{b}{B} = \occ{b}{A} \bmod 2$ for all $b$. So $f\left(\nf{\Xor{A}}\right) = f\left(\Xor{B}\right) = \nf{\Xor{f(B)}} = \Xor{}C$ for a \emph{set} $C$ such that $\occ{c}{C} = \occ{c}{f(B)} \bmod 2$ for all $c$. 
    
    For any $a \in A$, we have that $\occ{a}{A\setminus{B}} = \occ{a}{A} - \occ{a}{B} = \occ{a}{A} - (\occ{a}{A} \bmod 2)$ is an even number. Suppose, for a given $c$, that $\{a_{1}, \ldots, a_{n}\}$ is the \emph{set} of all $a \in A\setminus{B}$ s.t.\ $f(a) = c$. Then we have the following. 
    {
        \arraycolsep=2pt\def\arraystretch{1.1}
        \begin{longtable}{CL}
              & \occ{c}{f(A)} \bmod 2 \\
            = & \occ{c}{f(B)} \bmod 2 + \occ{c}{f(A\setminus{B})} \bmod 2 \\
            = & \occ{c}{f(B)} \bmod 2 + \occ{a_{1}}{A\setminus{B}} \bmod 2 + \cdots + \occ{a_{n}}{A\setminus{B}} \bmod 2 \\
            = & \occ{c}{f(B)} \bmod 2 + 0 + \cdots + 0 \\
            = & \occ{c}{f(B)} \bmod 2. \\
        \end{longtable}
    }
    From this it follows that $c \in C$ iff $\occ{c}{f(A)}$ is odd. In other words, $\Xor{C} = \nf{\Xor{f(A)}}$. Thus we have shown that $f\left(\nf{\Xor{A}}\right) = \nf{\Xor{f(A)}}$.
\end{proof} 

\begin{corollary}\label{cor:normintsub}
    For any multiset of standard normalized terms $\{t_{1}, \ldots, t_{n}\}$, we have: 
    \begin{enumerate}
        \item $\nf{\nf{t_{1} \xor \cdots \xor t_{n}}\,\intsub} = \nf{\nf{t_{1}\intsub} \xor \cdots \xor \nf{t_{n}\intsub}}$.
        \item $\zap{\nf{t_{1} \xor \cdots \xor t_{n}}} = \nf{\zap{t_{1}} \xor \cdots \xor \zap{t_{n}}}$.
    \end{enumerate}
\end{corollary}
\begin{proof}
    \phantom{a}
    \begin{enumerate}
        \item Note that $\nf{(t_{1} \xor \cdots \xor t_{n})\,\intsub} = \nf{t_{1}\intsub \xor \cdots \xor t_{n}\intsub} = \nf{\nf{t_{1}\intsub} \xor \cdots \xor \nf{t_{n}\intsub}}$, so the function $f \coloneq \lambda{t}\cdot{}\nf{t\intsub}$ satisfies the condition in Lemma~\ref{lem:funcnormA}. Applying the lemma, we get 
        \[ 
            \nf{\nf{t_{1} \xor \cdots \xor t_{n}}\,\intsub} = \nf{\nf{t_{1}\intsub} \xor \cdots \xor \nf{t_{n}\intsub}}. 
        \]
        \item By definition, $\zap{t_{1} \xor \cdots \xor t_{n}} = \nf{\zap{t_{1}} \xor \cdots \xor \zap{t_{n}}}$, so the function $f \coloneq \lambda{t}\cdot{}\zap{t}$ satisfies the condition in Lemma~\ref{lem:funcnormA}. Applying the lemma, we get the desired result.
        \qedhere
    \end{enumerate}
\end{proof}

\subsection{Typed terms in proofs}\label{subsec:typed-in-proofs}

\begin{lemma}\label{lem:nfintsub-derive}
    Let $X \cup \{t\}$ be a finite set of terms, and let $\pi$ be a proof of $X \vdash t$. There is a proof $\pi'$ of $\nf{X\intsub} \vdash \nf{t\intsub}$, which is the same proof as $\pi$ except that each node is labelled $(\nf{s\intsub}, \rnrule)$ instead of $(s,\rnrule)$. 
\end{lemma}
\begin{proof}
    Let $\pi$ and $\pi'$ be as in the statement of the lemma. We need to show that $\pi'$ is indeed a proof. For that we need to show that if $s$ is obtained from premises $s_{1}, \ldots, s_{n}$ via the rule $\rnrule$, then $\nf{s\intsub}$ can be obtained from $\nf{s_{1}\intsub}, \ldots \nf{s_{n}\intsub}$. We just show it for a few representative rules. 
    \begin{description}
        \item[$\rnrule = \rnsenc$:] $\nf{\sencfn{u}{v}\intsub} = \nf{\sencfn{u\intsub}{v\intsub}} = \sencfn{\nf{u\intsub}}{\nf{v\intsub}}$, so $\rnsenc$ applied on $\nf{u\intsub}$ and $\nf{v\intsub}$ will produce $\nf{\sencfn{u}{v}\intsub}$, as desired. 
        \item[$\rnrule = \rnsdec$:] As seen above, $\nf{\sencfn{u}{v}\intsub} = \sencfn{\nf{u\intsub}}{\nf{v\intsub}}$, so $\rnsdec$ applied on $\nf{\sencfn{u}{v}\intsub} $ and $\nf{v\intsub}$ will produce $\nf{u\intsub}$, as desired.
        \item[$\rnrule = \rnxor$:] In this case $s = \nf{s_{1} \xor \cdots \xor s_{n}}$. From Corollary~\ref{cor:normintsub}, $\nf{s\intsub} = \nf{\nf{s_{1}\intsub} \xor \cdots \xor \nf{s_{n}\intsub}}$, and so $\rnxor$ applied on $\nf{s_{1}\intsub}, \ldots, \nf{s_{n}\intsub}$ will produce $\nf{s\intsub}$, as desired.
        \qedhere
    \end{description}
\end{proof}

\begin{corollary}\label{cor:whconst-proofs}
    If $X \cup \{t\} \subseteq \constst$ and $X \DYderives t$, there is a proof $\pi^{*}$ of $\nf{X\intsub} \vdash \nf{t\intsub}$ such that $\termsof(\pi^{*}) \subseteq \nfintconst$. Furthermore, $\concof(\delta^{*}) \in \nfintnonvars$ for every subproof $\delta^{*}$ of $\pi^{*}$ whose last rule is $\rnpair, \rnaenc$ or $\rnsenc$. 
\end{corollary}
\begin{proof}
    Since $X \DYderives t$, there is a normal proof $\pi$ of $X \vdash t$. By the previous lemma, we get a proof $\pi^{*}$ of $\nf{X\intsub} \vdash \nf{t\intsub}$ s.t.\ each $u \in \termsof(\pi^{*})$ is of the form $\nf{r\intsub}$ for some $r \in \termsof(\pi)$. Since $\pi$ is normal and $X \cup \{t\} \subseteq \constst$, the subterm property guarantees that $\termsof(\pi) \subseteq \subterms{X \cup \{t\}} \subseteq \constst$. It follows that $\termsof(\pi^{*}) \subseteq \nfintconst$. 

    Consider a subproof $\delta^{*}$ as in the statement of the lemma. It is the translation of a subproof $\delta$ of $\pi$ with $\lrof(\delta) = \lrof(\delta^{*})$. It is clear that $\concof(\delta) \in \stnonvars$, and hence $\concof(\delta^{*}) \in \nfintnonvars$. 
\end{proof}

The next two lemmas hinge on the fact that a run is an interleaving of sessions of a protocol. 

\begin{lemma}\label{lem:sigmax-in-sigmaTp}\!\!\footnote{This is the only place where we use the fact that honest agent sends are derivable from previously received terms in the session, via Observation~\ref{obs:honestsends}. If we drop this feature of our model, then we can conclude that $t \in \subterms{\intsub(r_{k})}$, but not the stronger property that $t \in \subterms{\nf{\intsub(r_{k})}}$.  In~\cite{CKRT05}, they achieve this by imposing restrictions on the class of protocols admitted, and by a beautiful but involved proof (Lemma~13 in their paper).}
    Suppose $i \leq n$ and $t$ is a standard normalized term s.t.\ $t \in \subterms{\nf{X_{i}\intsub}}$. Then $t$ is typed or there is a $k \leq i$ such that $t \in \subterms{\nf{r_{k}\intsub}}$.
\end{lemma}
\begin{proof}
    Suppose $t \in \subterms{\nf{u\intsub}}$ for some $u \in X_{i}$. If $u \in X_{0}$, then since every term in $X_{0}$ is normalized and ground, $u$ and all its subterms are in $\stnonvars$, so $t = \nf{t\intsub} \in \nfintnonvars$, and hence $t$ is typed. So suppose $u \notin X_{0}$, i.e. $u = s_{j}$ for some $j \leq i$. So $t \in \subterms{\nf{s_{j}\intsub}}$. By Observation~\ref{obs:honestsends}, there is a set of ground terms $Y \subseteq \stnonvars$ s.t.\ $Y \cup \{r_{1}, \ldots, r_{j}\} \DYderives s_{j}$. Hence there is a normal proof $\pi$ of the same, and by Corollary~\ref{cor:whconst-proofs}, a proof $\pi^{*}$ of $U \vdash \nf{s_{j}\intsub}$ with $\termsof(\pi^{*}) \subseteq \nfintconst$, where $U = Y \cup \{\nf{r_{1}\intsub}, \ldots, \nf{r_{j}\intsub}\}$.  

    Since $t \in \subterms{\concof(\pi^{*})}$, we can consider a minimal subproof $\delta$ s.t.\ $t \in \subterms{\concof(\delta)}$. If $\lrof(\delta) = \rnax$, then $t \in \subterms{U}$, so either $t \in \subterms{Y} \subseteq \nfintnonvars$ (and hence $t$ is typed), or $t \in \subterms{\nf{r_{k}\intsub}}$ for some $k \leq j \leq i$. Note that $\lrof(\delta)$ cannot be any other destructor rule or an $\rnxor$ rule, since in those cases, $t \in \subterms{\concof(\delta')}$ for a proper subproof $\delta'$ of $\delta$, contradicting the minimality of $\delta$. So it has to be that $\lrof(\delta) \in \{\rnpair, \rnsenc, \rnaenc\}$ and $t = \concof(\delta)$. By Corollary~\ref{cor:whconst-proofs}, we conclude that $t \in \nfintnonvars$. 
    
    Thus in all cases we see that $t$ is typed or that $t \in \subterms{\nf{\intsub(r_{k})}}$ for some $k \leq i$. 
\end{proof}

\begin{lemma}\label{lem:earlier-proof}
    Suppose $i \leq n$, $t$ is a standard normalized term such that $\nf{X_{i}\intsub} \DYderives t$ via a normal proof $\pi$ ending in a destructor rule. Then $t$ is typed or there is an $\ell < i$ such that $\nf{X_{\ell}\intsub} \DYderives t$.
\end{lemma}
\begin{proof}
  	Since $\pi$ ends in a destructor rule, $t \in \subterms{\nf{X_{i}\intsub}}$. By Lemma~\ref{lem:sigmax-in-sigmaTp}, either $t$ is typed, or $t$ is untyped and there is an $i' \leq i$ such that $t \in \subterms{\nf{r_{i'}\intsub}}$. In the latter case, let $j$ be the earliest index such that $t \in \subterms{\nf{r_{j}\intsub}}$. Now $\nf{X_{j-1}\intsub} \DYderives \nf{r_{j}\intsub}$ via a normal proof $\rho$. Consider a minimal subproof $\chi$ of $\rho$ such that $t \in \subterms{\concof(\chi)}$. (There is at least one such subproof, namely $\rho$.) If $\chi$ ends in a destructor, then $\concof(\chi) \in \subterms{\nf{X_{j-1}\intsub}}$, and hence $t \in \subterms{\nf{X_{j-1}\intsub}}$. But by Lemma~\ref{lem:sigmax-in-sigmaTp}, there must be a $k \leq j-1$ such that $t \in \subterms{\nf{r_{k}\intsub}}$, contradicting the fact that $j$ is the earliest such index. So $\chi$ ends in a constructor rule. If $t \neq \concof(\chi)$, then either $\concof(\chi)$ is standard and $t$ is a proper subterm, or $\concof(\chi)$ is non-standard and $t$ is a subterm of a proper factor. In either case, it follows follows that $t \in \subterms{\concof(\chi')}$ for a proper subproof $\chi'$ of $\chi$. But this cannot be, since $\chi$ was assumed to be a minimal proof whose conclusion has $t$ as a subterm. Thus, $t = \concof(\chi)$ and $\chi$ is a proof of $\nf{X_{j-1}\intsub} \vdash t$ (and we choose our $\ell$ to be $j-1$, which is strictly less than $i$, since $j \leq i' \leq i$).           
\end{proof}

\subsection{Typed proofs}\label{subsec:typedproofs}

\begin{definition}[Typed derivations]\label{def:typed-proofs}
    A derivation $\pi$ is \defemph{typed} if for each subproof $\pi'$, either $\lrof(\pi')$ is a constructor rule, or $\concof(\pi')$ is typed, or $\concof(\pi')$ is non-standard.
\end{definition}

\begin{theorem}\label{thm:rustur}
	For all $t$ and all $i \in \embrace{0,\ldots,n}$, if $\nf{X_{i}\intsub} \DYderives t$, then there is a typed normal proof of $\nf{X_{i}\intsub} \vdash t$.
\end{theorem}
\begin{proof}
    Assume the theorem holds for all $j < i$. We show how to transform any normal proof $\pi$ of $\nf{X_{i}\intsub} \vdash t$ ending in rule $\rnrule$ into a typed normal proof $\pi^{*}$ of the same by induction on the structure of $\pi$. 
    \begin{description}
        \item[$\rnrule = \rnax$:] In this case, and in later cases where $\rnrule$ is a destructor, if $t$ is standard and non-typed, we argue as follows. By Lemma~\ref{lem:earlier-proof}, there is a $j < i$ and a proof $\rho$ of $\nf{X_{j}\intsub} \vdash t$. Since the statement of the theorem holds for all $j < i$, there is a typed normal proof $\rho^{*}$ of $\nf{X_{j}\intsub} \vdash t$, and it is also a proof of $\nf{X_{i}\intsub} \vdash t$. We define $\pi^{*} \coloneq \rho^{*}$. 
        
        If $t$ is non-standard or typed, it follows that $\pi$ itself is typed normal, and we just define $\pi^{*} \coloneq \pi$.
        \item[$\rnrule$ is a constructor:] We can find typed normal equivalents for all immediate subproofs, and apply the same constructor rule to get the desired $\pi^{*}$. 
        
        \item[$\rnrule$ is a destructor other than $\rnxord$:] If $t$ is standard and non-typed, we define $\pi^{*}$ as in the $\rnax$ case. 
        
        Otherwise, let $\pi_{1}$ and $\pi_{2}$ be immediate subproofs of $\pi$, with $\concof(\pi_{1}) = s$, and $t$ an immediate subterm of $s$. We can find typed normal equivalents $\pi^{*}_{1}$ and $\pi^{*}_{2}$. If $\pi^{*}_{1}$ ends in a constructor, then we choose $\pi^{*}$ to be the immediate subproof of $\pi^{*}_{1}$ s.t.\ $\concof(\pi^{*}) = t$. 
        
        If $\pi^{*}_{1}$ does not end in a constructor, $s \in \nfintnonvars$, and hence $t \in \nfintconst$. Since $t$ is either non-standard or typed, we obtain a typed normal $\pi^{*}$ by applying $\rnrule$ on $\pi^{*}_{1}$ and $\pi^{*}_{2}$. 
        
        \item[$\rnrule = \rnxord$:] Again, if $t$ is standard and non-typed, we define $\pi^{*}$ as in the $\rnax$ case. 
        
        Otherwise we argue as follows. Suppose the immediate subproofs of $\pi$ are $\pi_{1}, \ldots, \pi_{n}$, with $\concof(\pi_{k}) = t_{k}$ for $k \leq n$. By IH, there are typed normal proofs $\pi^{*}_{1}, \ldots, \pi^{*}_{n}$ deriving $t_{1}, \ldots, t_{n}$. Since $\pi$ is normal, $t \neq t_{i}$ for any $i \leq n$, and $t_{i} \neq t_{j}$ for $i \neq j$. Each immediate subproof of each $\pi^{*}_{k}$ is also typed normal. Let $\{\chi_{1}, \ldots, \chi_{\ell}\}$ be the set of all subproofs $\pi^{*}_{k}$ whose last rule is not $\rnxor$, along with all immediate subproofs of every $\pi^{*}_{k}$ whose last rule is $\rnxor$. Note that each $\chi_{i}$ is typed normal, and its last rule is not $\rnxor$. If $\concof(\chi_{i}) = t$ for some $i \le \ell$, we define $\pi^{*} \coloneq \chi_{i}$, which is typed normal. Let $\chi$ be the proof obtained by applying $\rnxord$ on $\chi_{1}, \ldots, \chi_{\ell}$. We define $\pi^{*}$ to be the proof obtained by normalizing $\chi$ (the only proof transformation needed is to repeatedly remove pairs  $\chi_{i}, \chi_{j}$ where $i \neq j$ and $\concof(\chi_{i}) = \concof(\chi_{j})$).   
        
        So $\pi^{*}$ is normal by construction, has typed normal immediate subproofs, and its conclusion $t$ is either non-standard or typed, so $\pi^{*}$ is a typed normal proof of $\nf{X_{i}\intsub} \vdash t$, as desired.
        \qedhere 
    \end{description}
\end{proof}

\subsection{Small substitution \texorpdfstring{$\vintsub$}{}, and simulating derivations}\label{subsec:smallsub}
\begin{definition}\label{def:vintsub}
    $\vintsub$ is a substitution with $dom(\vintsub) = \dom(\intsub)$ and $\vintsub(x) \coloneq \zap{\intsub(x)}$ for all $x \in \dom(\intsub)$.
\end{definition}

\begin{lemma}\label{lem:sigmastart-zapt}
    For $t \in \constst$, $\nf{t\vintsub} = \zap{\nf{t\intsub}}$. 
\end{lemma}
\begin{proof}
    We prove this by induction on the structure of $t$.
    \begin{description}
        \item[$t = x \in \vars$:] Since $x \in \constst$, $x \in \dom(\intsub)$. From Definition~\ref{def:vintsub} and from the fact that $\intsub$ and $\vintsub$ are normalized, we have that $\nf{x\vintsub} = \nf{\vintsub(x)} = \vintsub(x) = \zap{\intsub(x)} = \zap{\nf{\intsub(x)}} = \zap{\nf{x\intsub}}$. 
        \item[$t = a \in \names$:] Note that $\zap{a} = a$ is normalized. Thus $\nf{a\vintsub} = \nf{a} = a = \zap{a} = \zap{\nf{a}} = \zap{\nf{a\intsub}}$. 
        \item[$t = \sencfn{u}{v}$:] Since $u, v \in \subterms{\constst} \subseteq \constst$, we have $\nf{u\vintsub} = \zap{\nf{u\intsub}}$ and $\nf{v\vintsub} = \zap{\nf{v\intsub}}$ by IH. Since $t \in \stnonvars$, $\nf{t\intsub} \in \nfintnonvars$, and hence neither $t\intsub$ nor $\nf{t\intsub} = \sencfn{\nf{u\intsub}}{\nf{v\intsub}}$ is zappable. So we have the following: 
        {
            \setlength\tabcolsep{2pt}\def\arraystretch{1.2}
            \begin{longtable}{CLp{2em}l}
                  & \nf{\sencfn{u}{v}\vintsub} & & \\ 
                = & \nf{\nf{\sencfn{u}{v}\vintsub}} & & since $\nf{r} = \nf{\nf{r}}$ \\ 
                = & \nf{\nf{\sencfn{u\vintsub}{v\vintsub}}} & & expanding $\sencfn{u}{v}\vintsub$ \\
                = & \nf{\sencfn{\nf{u\vintsub}}{\nf{v\vintsub}}} & & expanding $\nf{\sencfn{r}{s}}$ \\ 
                = & \nf{\sencfn{\zap{\nf{u\intsub}}}{\zap{\nf{v\intsub}}}} & & by IH \\
                = & \zap{\sencfn{\nf{u\intsub}}{\nf{v\intsub}}} & & the term is not zappable \\
                = & \zap{\nf{\sencfn{u\intsub}{v\intsub}}} & & \\ 
                = & \zap{\nf{\sencfn{u}{v}\intsub}} & & \\
            \end{longtable}
        } 
        A similar argument works when $t = \pairfn{u}{v}$ or $t = \aencfn{u}{\pkfn{k}}$. 
        \item[$t = t_{1} \xor \cdots \xor t_{n}$:] Since $\{t_{1}, \ldots, t_{n}\} \subseteq \subterms{\constst} \subseteq \constst$, we have $\nf{t_{i}\vintsub} = \zap{\nf{t_{i}\intsub}}$ for all $i \leq n$, by IH. Let $M = \Cup_{i \leq n} \factors(\nf{t_{i}\intsub}) = \{u_{1}, \ldots, u_{\ell}\}$. Note that each $u_{j}$ is normalized and standard. We have the following chain of equations. 
        {
            \setlength\tabcolsep{2pt}\def\arraystretch{1.2}
            \begin{longtable}{CLp{2em}l}
                & \nf{(t_{1} \xor \cdots \xor t_{n})\,\vintsub} & & \\ 
                = & \nf{t_{1}\vintsub \xor \cdots \xor t_{n}\vintsub} & & expanding $\vintsub$ \\ 
                = & \nf{\nf{t_{1}\vintsub} \xor \cdots \xor \nf{t_{n}\vintsub}} & & $\xor$ theory \\ 
                = & \nf{\ \zap{\nf{t_{1}\intsub}} \xor \cdots \xor \zap{\nf{t_{n}\intsub}}\ } & & by IH \\ 
                = & \nf{\zap{\Xor{}\factors(\nf{t_{1}\intsub})} \xor \cdots \xor \zap{\Xor{}\factors(\nf{t_{n}\intsub})}} & & $\nf{r} = \Xor{}\factors(\nf{r})$ \\
                = & \nf{\nf{\Xor{}\zap{\factors(\nf{t_{1}\intsub})}} \xor \cdots \xor \nf{\Xor{}\zap{\factors(\nf{t_{n}\intsub})}}} & & def. of zap for $\xor$-terms \\
                = & \nf{(\Xor{}\zap{\factors(\nf{t_{1}\intsub})}\,) \xor \cdots \xor (\Xor{}\zap{\factors(\nf{t_{n}\intsub})}\,)} & & $\xor$ theory \\
                = & \nf{\Xor{}\Cup_{i\leq{}n}\zap{\factors(\nf{t_{i}\intsub})}} & & $\xor$ theory \\
                = & \nf{\zap{u_{1}} \xor \cdots \xor \zap{u_{\ell}}} & & def. of the $u_{j}$s \\ 
                = & \zap{\nf{u_{1} \xor \cdots \xor u_{\ell}}} & & by Corollary~\ref{cor:normintsub} \\ 
                = & \zap{\nf{\Xor{}\Cup_{i\leq{}n} \factors(\nf{t_{i}\intsub})}} & & def. of the $u_{j}$s \\
                = & \zap{\nf{\nf{t_{1}\intsub} \xor \cdots \xor \nf{t_{n}\intsub}}} & & $\xor$ theory \\
                = & \zap{\nf{t_{1}\intsub \xor \cdots \xor t_{n}\intsub}} & & $\xor$ theory \\ 
                = & \zap{\nf{(t_{1} \xor \cdots \xor t_{n})\,\intsub}} & & un-expanding $\intsub$. \\ 
            \end{longtable}
        }
    \end{description}
    Thus for all $t \in \constst$, $\nf{t\vintsub} = \zap{\nf{t\intsub}}$, as desired.
\end{proof}

\begin{theorem}\label{thm:simulate}
	For $i \leq n$ and any term $t \in \constst$, if $\nf{X_{i}\intsub} \DYderives \nf{t\intsub}$ then $\nf{X_{i}\vintsub} \DYderives \nf{t\vintsub}$.
\end{theorem}
\begin{proof}
    By Theorem~\ref{thm:rustur}, there is a typed normal proof of $\nf{X_{i}\intsub} \vdash \nf{t\intsub}$. Since $X_{i} \subseteq \constst$, it follows from Lemma~\ref{lem:sigmastart-zapt} that $\nf{X_{i}\vintsub \cup \{t\vintsub\}} = \zap{\nf{X_{i}\intsub}} \cup \{\zap{\nf{t\intsub}}\}$. Thus it suffices to prove that $\zap{\nf{X_{i}\intsub}} \vdash \zap{\nf{t\intsub}}$. 

    We will prove a more general fact -- for any typed normal proof $\pi$ of $X \vdash t$ with $0 \in X$ and $t$ normalized, there is a proof of $\zap{X} \vdash \zap{t}$. We prove this by induction on the structure of $\pi$. Assume that for all proper subproofs $\delta$ of $\pi$, if $\concof(\delta) = r$, there is a proof of $\zap{X} \vdash \zap{r}$. Let $\lrof(\pi) = \rnrule$.
    \begin{description}
		\item[$\rnrule = \rnax$:] $t \in X$, and therefore $\zap{t} \in \zap{X}$. Thus $\zap{X} \DYderives \zap{t}$ by $\rnax$. 
		\item[$\rnrule$ is a constructor rule other than $\rnxor$:] Let $t = \func(t_{1}, t_{2})$ and let $\pi_{1}, \pi_{2}$, with $\concof(\pi_{i}) = t_{i}$, be the immediate subproofs of $\pi$. By IH, there are proofs $\delta_{1}$ and $\delta_{2}$ of $\zap{X} \vdash \zap{t_{1}}$ and $\zap{X} \vdash \zap{t_{2}}$. If $t$ is zappable, then $\zap{t} = 0 \in \zap{X}$ ($0 \in X$, so $0 = \zap{0} \in \zap{X}$), and we have $\zap{X} \DYderives \zap{t}$ using $\rnax$. If $t$ is not zappable, then $\zap{t} = \nf{\func(\zap{t_{1}}, \zap{t_{2}})} = \func(\zap{t_{1}}, \zap{t_{2}})$, and we can apply $\rnrule$ on $\delta_{1}$ and $\delta_{2}$ to get $\zap{X} \DYderives \zap{t}$.
		\item[$\rnrule$ is a destructor other than $\rnxor$:] Let the immediate subproofs of $\pi$ be $\pi_{1}$ and $\pi_{2}$, deriving $t_{1}$ and $t_{2}$ respectively, with the standard term $t_{1}$ being the major premise, and $t$ an immediate subterm of $t_{1}$. Since $\pi$ is typed normal, $\pi_{1}$ is also typed and ends in a destructor, so by Definition~\ref{def:typed-proofs}, $t_{1} \in \nfintnonvars$. By definition, $t_{1}$ is not zappable, so $\zap{t_{1}}$ has the same outermost operator as $t_{1}$. By IH, there are proofs $\delta_{1}$ and $\delta_{2}$ of $\zap{X} \vdash \zap{t_{1}}$ and $\zap{X} \vdash \zap{t_{2}}$. Since $\zap{t_{1}}$ is not atomic, we can apply the destructor $\rnrule$ on $\delta_{1}$ and $\delta_{2}$ to get $\zap{X} \DYderives \zap{t}$. 
		\item[$\rnrule = \rnxor$:] Let the immediate subproofs of $\pi$ be $\pi_{1}, \ldots, \pi_{n}$, deriving $t_{1}, \ldots, t_{n}$ respectively. Note that each $t_{i}$ is normalized and $\pi$ itself derives $t = \nf{t_{1} \xor \cdots \xor t_{n}}$. By IH, there are proofs $\delta_{1}, \ldots, \delta_{n}$, each $\delta_{i}$ deriving $\zap{X} \vdash \zap{t_{i}}$. Applying $\rnxor$ on these gives a proof $\delta$ of $\zap{X} \vdash \nf{\zap{t_{1}} \xor \cdots \xor \zap{t_{n}}}$. Letting $M = \factors(t_{1} \xor \cdots \xor t_{n})$, we have the following chain of equations.
        {
            \arraycolsep=2pt\def\arraystretch{1.1}
            \begin{longtable}{CCLl}
                \zap{t} & = & \zap{\nf{t_{1} \xor \cdots \xor t_{n}}} & \\
                & = & \zap{\nf{\Xor{M}}} & all elements of $M$ are standard \& normalized \\
                & = & \nf{\Xor{\zap{M}}} & by Corollary~\ref{cor:normintsub} \\ 
                & = & \nf{\nf{\Xor{\zap{\factors(t_{1})}}} \xor \cdots \xor \nf{\Xor{\zap{\factors(t_{n})}}}} & $\xor$ theory \\
                & = & \nf{\zap{t_{1}} \xor \cdots \xor \zap{t_{n}}} & def. of zap for $\xor$ terms \\ 
            \end{longtable}
        }
        Thus we see that $\delta$ is indeed a proof of $\zap{X} \vdash \zap{t}$, as desired, and this completes the proof.
        \qedhere
	\end{description}
\end{proof}

Recall that for any substitution $\lambda$, $\size{\lambda} = |\subterms{\{\lambda(x) \mid x \in \dom(\lambda)\}}|$. The following lemma is useful in showing that $\vintsub$ is \defemph{small}, in the sense that $\size{\vintsub} \leq |\constst|$.

\begin{lemma}\label{lem:stzapnfintsub}
    For every $u \in \zap{\nfintconst}$, $\subterms{u} \subseteq \zap{\nfintconst}$. 
\end{lemma} 
\begin{proof}
    The proof is by induction on the structure of $u$. The following cases arise.
    \begin{description}
        \item[$u = a \in \names$:] If $a \in \zap{\nfintconst}$, $\subterms{a} = \{a\} \subseteq \zap{\nfintconst}$. 
        
        \item[$u = \func(v,w)$ is standard:] Suppose $t \in \constst$ s.t.\ $u = \zap{\nf{t\intsub}}$. Therefore $\nf{u} = u$ and $\zap{u} = u$. Suppose $u \notin \zap{\nfintnonvars}$. This means that $u$ is zappable. But $u$ is already of the form $\zap{\nf{t\intsub}}$ for some $t \in \constst$, so it has to be the case that $u = \zap{u} = 0$, which is a contradiction.
        
        So there is some standard $p \in \stnonvars$ s.t.\ $u = \zap{\nf{p\intsub}}$. Then $p = \func(r,s)$ for some $r, s \in \constst$. Since $v = \zap{\nf{r\intsub}}$ and $w = \zap{\nf{s\intsub}}$, by IH, we have $\subterms{v} \cup \subterms{w} \subseteq \zap{\nfintconst}$. Since $u \in \zap{\nfintconst}$ too, we have that $\subterms{u} \subseteq \zap{\nfintconst}$. 
        
        \item[$u = u_{1} \xor \cdots \xor u_{n}$ where each $u_{i}$ is standard:] Suppose $t \in \constst$ s.t.\ $u = \zap{\nf{t\intsub}}$. Therefore $\nf{u} = \zap{u} = u$, and for each $i \leq n$, $\nf{u_{i}} = \zap{u_{i}} = u_{i}$. Consider $u_{i}$ for any $i \leq n$. If $u_{i} \notin \nfintconst$, we have $\nf{u_{i}} = u_{i} \notin \nfintnonvars$, which means that $u_{i}$ is zappable. Hence $u_{i} = \zap{u_{i}} = 0 \in \nfintconst$, which is a contradiction. Hence for every $i \leq n$, there is some $u_{i} \in \zap{\nfintconst}$. Since each $u_{i}$ is a proper subterm of $u$, we can apply IH and conclude that $\subterms{u_{i}} \subseteq \zap{\nfintconst}$ for each $i \leq n$. Since $u \in \zap{\nfintconst}$ too, we have that $\subterms{u} \subseteq \zap{\nfintconst}$. 
        \qedhere
    \end{description}
\end{proof}

\begin{theorem}\label{thm:vintsub-small}
    $\vintsub$ is small.
\end{theorem}
\begin{proof}
    For every $x \in \dom(\intsub) = \dom(\vintsub)$, $\vintsub(x) = \zap{\intsub(x)} = \zap{\nf{\intsub(x)}} \in \zap{\nfintconst}$, by definition of $\vintsub$ and the fact that $\intsub$ is normalized. Applying Lemma~\ref{lem:stzapnfintsub}, we have $\subterms{\vintsub(x)} \subseteq \zap{\nfintconst}$ for all $x \in \dom(\vintsub)$, and hence $\subterms{\{\vintsub(x) \mid x \in \dom(\vintsub)\}} \subseteq \zap{\nfintconst}$. But then it is clear that 
    \[
        \size{\vintsub} = |\subterms{\{\vintsub(x) \mid x \in \dom(\vintsub)\}}| \leq |\zap{\nfintconst}| \leq |\nfintconst| \leq |\intconst| \leq |\constst|. \qedhere
    \]  
\end{proof}

\subsection{An NP decision procedure for \texorpdfstring{$B$-}{}bounded insecurity}
Fix $K$. Given a protocol $\prot$, guess a set $\sessions$ of $K$ pairwise coherent sessions of $\prot$. Each session is a role along with a substitution $\honsub$ for the agent variables s.t.\ honest agent sends are derivable after applying $\honsub$. Notice that each such $\honsub$ is small, since $\honsub(x) \in \names \cup \vars$, and so $\size{\honsub} \leq |\constst|$. There are polynomially many derivability checks, each of which can be performed in time polynomial in $|\constst|$. We then check the coherence of this set of sessions. 

Finally, we guess a run $(\xi, \vintsub)$, where $\xi$ is a prefix of an interleaving of the sessions in $\sessions$, and $\vintsub$ is a small substitution. The length $n$ of $\xi$ is also polynomial in the input size, and therefore there are polynomially many derivability checks for the intruder, each solvable in polynomial time. Thus we have an NP procedure to check if a given protocol has an attack consisting of $K$ sessions.

\end{document}